\documentclass[conference]{IEEEtran}
\usepackage{amsmath}
\usepackage{amssymb}
\usepackage{graphicx}
\usepackage{enumitem}
\usepackage{fancyhdr}
\usepackage{bbm}
\usepackage{physics}
\usepackage[style=ieee]{biblatex}
\usepackage{amsthm,amsfonts,paralist,mathtools,physics}
\usepackage[hidelinks]{hyperref}
\usepackage{color,dsfont} 
\makeatletter
\renewcommand*\env@matrix[1][*\c@MaxMatrixCols c]{%
  \hskip -\arraycolsep
  \let\@ifnextchar\new@ifnextchar
  \array{#1}}
\makeatother

\usepackage{listings}
\newtheorem{theorem}{Theorem}

\newcommand{\floor}[1]{\lfloor #1 \rfloor}

\newcommand{\tensor}{\otimes}
\newcommand{\ceil}[1]{\left\lceil #1 \right\rceil}
\newcommand{\GF}{\text{GF}}
\newcommand{\qcode}[1]{\left[\!\left[ #1 \right]\!\right]}

\newcommand{\st}[2]{\genfrac{[}{]}{0pt}{}{#1}{#2}}

\title{Robust Syndrome Extraction via BCH Encoding}
\author{Eren Guttentag$^{1,2}$, Andrew Nemec$^{1,2}$, and Kenneth R. Brown$^{1,2,3,4}$\\\small%
	$^{1}$Duke Quantum Center, Duke University, Durham, NC 27701, USA\\%
	$^{2}$Department of Electrical and Computer Engineering, Duke University, Durham, NC 27708, USA\\%
    $^{3}$Department of Physics, Duke University, Durham, NC 27708, USA\\%
    $^{4}$Department of Chemistry, Duke University, Durham, NC 27708, USA
}
\begin{document}
	\maketitle
	\begin{abstract}
  Quantum data-syndrome (QDS) codes are a class of quantum error-correcting codes that protect against errors both on the data qubits and on the syndrome itself via redundant measurement of stabilizer group elements. One way to define a QDS code is to choose a syndrome measurement code, a classical block code that encodes the syndrome of the underlying quantum code by defining additional stabilizer measurements.  
  
  We propose the use of primitive narrow-sense BCH codes as syndrome measurement codes. 
  We show that these codes asymptotically require $O(t\log\ell)$ extra measurements, where $\ell$ is the number of stabilizers generators of the quantum code and $t$ is the number of errors corrrected by the BCH code.

    Previously, the best known general method of constructing QDS codes out of quantum codes requires $O(t^3\log\ell)$ extra measurements. As the number of additional syndrome measurements is a reasonable metric for the amount of additional time a general QDS code requires, we conclude that our construction protects against the same number of syndrome errors with significantly less time overhead. 
	    
	\end{abstract}
	\section{Introduction}
    The ability to accurately detect, identify, and correct errors is essential to building functional and scalable quantum computers. This is typically achieved through the use of quantum stabilizer codes, which are defined by their stabilizer group. The measurement of the group generators produces a binary syndrome that indicates the locations of errors. These measurements involve several multi-qubit gates, which can introduce errors on the qubits involved and corrupt the measurement outcome.

    For syndrome fault tolerance, we can measure additional stabilizer group elements to add redundancy. One common way to do this is to repeatedly measure the same set of stabilizers \cite{Shor1996,DiVincenzo1996}. However, the syndrome can be protected much more efficiently with the use of quantum data-syndrome (QDS) codes. Introduced over a series of papers by Fujiwara \cite{Fujiwara2014, Fujiwara2015} and Ashikhmin, Lai, and Brun \cite{Ashikhmin2014, Ashikhmin2016, Ashikhmin2020}, QDS codes simultaneously encode quantum information and protect against syndrome errors. Recent work includes decoding QDS codes \cite{Kuo2021, Raveendran2022}, connections to 2-designs \cite{Premakumar2021}, and extenstions to quantum convolutional \cite{Zeng2019} and subsystem codes \cite{Nemec2023}. QDS codes are also closely related to single-shot quantum error correction \cite{Bombin2015,Campbell2019}.

    We propose a way of constructing quantum data-syndrome codes using primitive narrow-sense BCH codes. We will also show that such a construction that protects against $t$ syndrome errors requires $O(t\log(n-k))$ additional measurements, which is a significant improvement over other ways of constructing these codes. Recently, BCH codes have also been used to design good flag fault-tolerant syndrome extraction schemes \cite{Anker2022}.

    In this paper we will look at a phenomenological error model, in which individual gates of a stabilizer have a chance of causing an error on the syndrome bit. This model does not consider hook errors propagated onto a circuit due to gate errors, but is intended to be a stepping-stone to exploring a full circuit model in the future.

	\section{Background}
 
    \subsection{Stabilizer codes}
    
	Pauli operators on $n$ qubits are $n$-fold tensor products of Pauli matrices $P_0=I=\smqty[1&0\\0&1],P_1=X=\smqty[0&1\\1&0],P_2=Y=\smqty[0&-i\\i&0],P_3=Z=\smqty[1&0\\0&-1]$, of the form $$i^{c}\cdot P_{a_0}\tensor P_{a_1}\tensor\cdots\tensor P_{a_{n-1}},$$ with $a_i,c\in\{0,1,2,3\}$. For simplicity's sake we will omit tensor products from our notation; so the three-qubit operator $X\tensor Y\tensor I$ becomes $XYI$. These operators form the group $\mathcal{P}^n$.
    
    An $\qcode{n,k,d}$ stabilizer code encoding $k$ logical qubits into $n$ physical qubits is defined by its $\ell:=n-k$ independent \textit{stabilizer} generators, $\{g_1,g_2,\dots,g_\ell\}\subset\mathcal{P}^n$. These operators generate the stabilizer group $\mathcal{S}\subset\mathcal{P}^n$ of the code, which is commutative, does not contain $-I^{\tensor n}$, and has order $2^{\ell}$. Elements of the stabilizer group fix the states in the quantum code: for any stabilizer $g_i\in\mathcal{S}$ and any codeword $\ket{\psi}$ of the associated quantum code, $g_i\ket\psi=\ket\psi$.
    
    A stabilizer on $n$ qubits can be seen as a length-$n$ vector with elements in $\GF(4)$, under the homomorphism $\tau:\mathcal{P}\to\GF(4)$ that maps $I\to0,X\to1,Y\to\bar\omega=1+\omega,Z\to\omega$ and ignores any global phase of $\pm1,\pm i$. This can be naturally extended to send operators on $n$ qubits from $\mathcal{P}^n$ to $\GF(4)^n$ \cite{Ashikhmin2020}. Under this homomorphism, multiplication of stabilizers corresponds to bitwise addition in $\GF(4)^n$. Throughout this paper we will use $g\in\mathcal{P}^n$ to refer to a stabilizer and $\mathbf{g}\in\GF(4)^n$ to refer to the corresponding length-$n$ $\GF(4)$ vector $\tau(g)$. For instance, the 3-qubit stabilizer $g=IXY$ is equivalently represented by the length-3 $GF(4)$ vector $\mathbf{g}=\tau(IXY) =01\bar\omega$.
    
    In general, an error $E$ on an $\qcode{n,k,d}$ quantum code can be represented as an element of $\mathcal{P}^n$. The measurement of its set of $\ell$ stabilizer operators gives as an output a length-$\ell$ binary vector called the syndrome $\mathbf{s}=(s_1,s_2,\dots,s_\ell)$. The $i$-th syndrome bit $s_i$, corresponding to stabilizer operator $g_i$, is 0 if $E$ and $g_i$ commute, and 1 if they anticommute. For vectors $\mathbf{x},\mathbf{y}\in\GF(4)^n$ with elements $x_i,y_i\in\GF(4)$, the analagous function is the \textit{trace inner product} $\mathbf{x}\star\mathbf{y}:\GF(4)^n\cross\GF(4)^n\to \GF(2)$:
    $$\mathbf{x}\star\mathbf{y}=\sum_{i=1}^n(x_i\bar{y_i}+\bar{x_i}y_i),$$
    where $\bar{0}=0,\bar{1}=1,\bar{\omega}=1+\omega,\overline{1+\omega}=\omega$, and multiplication is typical in $\GF(4)$ \cite{Calderbank1997}. Essentially, the $i$-th element of the sum is 0 if $x_i=0,\,y_i=0$, or $x_i=y_i$, and 1 otherwise. If an odd number of summands are 1, the trace will be 1, and if an even number are 1, the trace will be 0. Therefore the trace inner product is 0 when the Pauli operators represented by $\mathbf{x}$ and $\mathbf{y}$ commute, and is 1 when they anticommute. 

	\subsection{Quantum Data-Syndrome Codes}

    We assume that Shor-style syndrome extraction \cite{Shor1996,DiVincenzo1996} is used for measurement, in which a weight-$w$ stabilizer can be measured fault-tolerantly using $w$ single-qubit measurements. This is done using transversal gates and a $w$-qubit ancilla cat state to prevent a single error on an ancilla qubit from propagating to more than one data qubit.
    
    In this paper, we focus on the phenomenological model where errors can either occur on the data qubits or on the syndrome bits. We only consider \textit{measurement errors}---errors that result in a single syndrome bit-flip, equivalent to an error on the ancilla qubit after all transversal gates have already occurred, with no propagation to data qubits. Let $p_{m}$ be the probability of a single-qubit measurement error. Then the probability of incorrectly measuring a stabilizer generator $S_{i}$ of weight $w_{i}$ is given by
    \begin{equation}
    p_{err}\!\left(S_{i}\right)=\sum\limits_{j\text{ odd}}\binom{w_{i}}{j}p_{m}^{j}\left(1-p_{m}\right)^{w_{i}-j}.
    \end{equation}

    If any bits of the syndrome are flipped during measurement, the resulting erroneous syndrome vector $\hat{\mathbf{s}} =\mathbf{s}+\hat{\mathbf{e}}$ may suggest an incorrect series of gates to restore the state. Some stabilizer codes have a choice of generators that allow them to correct either a single data error or a single syndrome error \cite{Fujiwara2015, Nemec2023}, but in general, we need to perform additional stabilizer measurements to add redundancy against syndrome errors. Shor accomplished this by repeatedly measuring each stabilizer generator multiple times \cite{Shor}. However, this can be achieved more efficiently by measuring an overdetermined set of stabilizer elements. We call a quantum code with $r$ extra measurements beyond the $\ell=n-k$ needed an $\qcode{n,k,d:r}$ quantum data-syndrome (QDS) code if it can correct up to a combined $\left\lfloor\left(d-1\right)/2\right\rfloor$ data and syndrome errors.
	
	A reasonable metric for the amount of extra time a QDS code takes is the number of extra syndrome measurements that take place. Assuming elements of the syndrome group are of roughly equal weight, the measurement of each syndrome bit involves a similar number of gates to be performed, and thus takes a similar amount of time \cite{Delfosse2022}. This assumes that measurements are done sequentially; the ability to measure in parallel could produce further improvement. Because of this, the most efficient QDS codes are those that correct a certain number of errors with the fewest additional stabilizer measurements. This is one reason that encoding syndrome information by simply repeating stabilizer measurements is inefficient---an $m$-fold repetition of $\ell$ stabilizer measurements that can correct up to $\floor{\frac{m-1}{2}}$ syndrome errors requires $O(m\ell)$ additional measurements. 
	 
	A general framework for designing a QDS code with a desired total distance $2t_c+1$ was proposed by Fujiwara \cite{Fujiwara2015}. It involves \textit{$s$-detection parity check matrices} ($s$-DPMs), which are parity-check matrices for codes that can \textit{detect} up to $s$ errors (without necessarily being able to correct or identify them) \cite{Fujiwara2015}. Specfically, an $s$-DPM$(m,w)$ is a binary $m\times w$ matrix such that any $m\times s$ submatrix contains a row of odd weight. In \cite[Theorem 3.5]{Fujiwara2015}, Fujiwara showed that a $2i$-DPM$(m,w)$ exists so long as:
    $$m\geq \ceil{\log_2\left(\binom{w}{2i}-\binom{w-2i}{2i}\right)+\log_2e}.$$

    Fujiwara uses these parity check matrices to construct a QDS code from a $\qcode{n,k,2t+1}$ stabilizer code that can protect against all errors such that the sum of the errors on syndrome bits and data qubits is $t$ or fewer, and at most $t_c$ of those occur on the syndrome bits. The resulting stabilizer parity check matrix has $|S|$ rows where:
    $$|S|=n-k+2t_c+\sum_{i=1}^{t_c}(2t_c-2i+1)m_i,$$
    $$m_i=\ceil{\log_2\left(\binom{n-k}{2i}-\binom{n-k-2i}{2i}\right)+\log_2e}.$$

    \begin{theorem}
    The construction in \cite{Fujiwara2015} requires $O(t_c^3\log\ell)$ additional stabilizer measurements.
    \end{theorem}
    \begin{proof}

        As we have $\ell=n-k$, then we can express $m_i$ in terms of $\ell$ as:

        $$m_i=\ceil{\log_2\left(\binom{\ell}{2i}-\binom{\ell-2i}{2i}\right)+\log_2e}.$$

        As we are interested in the behavior of this code asymptotically, we look at the behavior for larger $\ell$. The binomial coefficient $\binom{n}{k}$ can in general be expressed as the polynomial:
        \begin{equation}
            \binom{n}{k}=\sum_{i=0}^ks(k,i)\frac{n^i}{k!},
        \end{equation}
    where $s(k,i)=(-1)^{k-i}\st{k}{i}$ are the Stirling numbers of the first kind \cite{Abramowitz1972}. There are several special values of Stirling numbers of the first kind; most valuable to us are $s(n,n)=\st{n}{n}=1$ and $s(n,n-1)=-\st{n}{n-1}=-\binom{n}{2}=-\frac{n(n-1)}{2}$. 
    
    Now consider the terms in these expansions with order greater than $2i-2$:
    
    $$\binom{\ell}{2i}=\st{2i}{2i}\frac{\ell^{2i}}{(2i)!}-\st{2i}{2i-1}\frac{\ell^{2i-1}}{(2i)^!}+O(\ell^{2i-2}),$$
    $$\binom{\ell-2i}{2i}=\st{2i}{2i}\frac{(\ell-2i)^{2i}}{(2i)!}-\st{2i}{2i-1}\frac{(\ell-2i)^{2i-1}}{(2i)!}+O(\ell^{2i-2}).$$

    Expanding the $(\ell-2i)^{2i}$ and $(\ell-2i)^{2i-1}$ terms, we see:
		$$(\ell-2i)^{2i} = \ell^{2i}-(2i)^2\ell^{2i-1}+O(\ell^{2i-2}),$$
		$$(\ell-2i)^{2i-1} = \ell^{2i-1}+O(\ell^{2i-2}),$$
    so in fact the $\ell^{2i}$ and $\ell^{2i-1}$ terms of the polynomials are (evaluating the Stirling numbers of the first kind):
    $$\binom{\ell}{2i}=\frac{\ell^{2i}}{(2i)!}-\frac{2i(2i-1)}{2}\frac{\ell^{2i-1}}{(2i)^!}+O(\ell^{2i-2}),$$
		$$\binom{\ell-2i}{2i} =\frac{\ell^{2i}}{(2i)!}-\frac{(2i)^2\ell^{2i-1}}{(2i)!}-\frac{2i(2i-1)}{2}\frac{\ell^{2i-1}}{(2i)!}+O(\ell^{2i-2}).$$
    Then we can see the difference between these two binomial coefficients will have a highest-order term in $\ell$ of the form $\ell^{2n-1}$:
    $$\binom{\ell}{2i}-\binom{\ell-2i}{2i}=\frac{2i\ell^{2i-1}}{(2i-1)!}+O(\ell^{2i-2}).$$

    We know the $m_i$ is at least $(2i-1)\log_2\ell$:
    $$m_i=\ceil{\log_2\left(\binom{\ell}{2i}-\binom{\ell-2i}{2i}\right)+\log_2e}$$
    $$\Downarrow$$
    $$m_i>\log_2\left(\frac{2i\ell^{2i-1}}{(2i-1)!}\right)=(2i-1)\log_2\ell+O(1).$$
	So the number of redundant stabilizer measurements required by this construction is at least $$2t_c+\sum\limits_{i=1}^{t_c}(2t_c-2i+1)(2i-1)\log_2\ell.$$
    This sum can actually be evaluated to a closed form; $\sum_{i=1}^{t_c}(2i-1)(2t_c-2i+1)\log_2\ell=\log_2\ell(\frac{2t^3+t}{3})$. So in terms of $t$ and $\ell$, the total number of additional stabilizers constructed is
    $$|S|-(n-k)=2t_c+\frac{\log_2\ell}{3}\left(2t_c^3+t_c\right).$$
    The dominant term of this is $\frac{2t_c^3\log_2\ell}{3}$, and $$\frac{2t_c^3\log_2\ell}{3}\in O(t_c^3\log\ell).$$

 \end{proof}
	
	\section{Syndrome Measurement Codes}
	
	In general, given a stabilizer code, it is nontrivial to choose a set of stabilizer generators such that the resulting QDS code has a good total minimum distance. Instead we make use of \textit{syndrome measurement (SM) codes}. A syndrome measurement code is a $[n_{S},\ell,2t_{S}+1]$ classical block code that defines an overdetermined set of $n_{S}$ stabilizer operators to be measured. This allows for a two-step decoding protocol that is simpler than simultaneously decoding syndrome and data errors, but can perform suboptimally in comparison \cite{Ashikhmin2020}. In the classical decoding step, the measured length-$n_C$ bit string is decoded with a decoder of the SM code. This results in a length-$\ell$ syndrome for the stabilizer code, which is then used to correct quantum errors in the second step. One advantage of using a SM code is that the number of correctable syndrome bit-flip errors is easy to dictate and independent from the minimum distance of the stabilizer code.

    If we have a stabilizer code $\qcode{n,k,d}$ whose syndrome is encoded in a $[n_C,\ell,d_C]$ classical code, then the overall minimum distance of the QDS code is $d'\geq\min(d,d_C)$, and it can correct up to a simultaneous $\floor{\frac{d-1}{2}}$ errors on the data qubits and $\floor{\frac{d_C-1}{2}}$ errors on the syndrome bits. This ability to control the distance of the SM code makes the codes particularly useful for systems with relatively high probability of measurement error. 
	
	\section{Protecting Syndromes with BCH Codes}

	We propose the use of primitive narrow-sense Bose-Chaudhuri-Hocquenghem (BCH) codes as SM codes to encode the syndrome bits. These codes are a class of cyclic binary codes of the form $[2^{m}-1, 2^{m}-R(m,t) - 1, 2t+1]$, where $R(m,t)\leq mt$ for a chosen $m,t\in\mathbb{N}$. The properties of a BCH code are defined by the degree of the least common multiple of certain irreducible polynomials \cite{Bose1960b,Bose1960a,Hocquenghem1959}. Any BCH code can be shortened to create a $[2^m-1-a,2^m-R(m,t)-1-a,2t+1]$ code. The generators of the shortened code are the codewords that are zero in their first $a$ bits. Note that because both the number of logical and data bits both decrease by $a$, this process does not change the difference between the number of logical and data bits. This means that whether we use a regular or shortened BCH code as an SM code, the number of additional stabilizers measured will still be $R(m,t)$.

    In order to encode a syndrome of an $\qcode{n,k,d}$ quantum code (with a $\ell\times 2n$ stabilizer matrix $H$) in a BCH code with distance $d_{S}=2t_{S}+1$, we choose $m$ to be the smallest integer such that $\ell\leq 2^m-mt_{S}-1$. Note that this means $2^{m-2}<\ell<2^m<4\ell$, and so $\log_2\ell$ is $O(m)$ and $m$ is $O(\log\ell)$. 
    
    We use this $m$ and our desired $t_{S}$ to find the corresponding $[2^{m}-1,2^m-R(m,t_S)-1,2t_S+1]$ BCH code, and if $2^{m}-R(m,t_{S})-1>\ell$ we shorten it by an appropriate amount so that it is a $[\ell+R(m,t_{S}),\ell,2t_{S}+1]$ code. This code has a $\ell\times (\ell+R(m,t_{S}))$ generator matrix $G_{B}$. We can use this matrix to generate the $(\ell+R(m,t_{S}))\times 2n$ stabilizer matrix for our $\qcode{n,k,d:R(m,t_{S})}$ QDS code, $H_{Q}:=G_{B}^TH$.

\begin{theorem}
    Using a BCH code or shortened BCH code that can correct up to $t$ errors as a SM code encoding $\ell$ syndrome bits requires $O(t\log\ell)$ additional stabilizer measurements.
\end{theorem}

\begin{proof}
    When we encode $\ell$ bits in a BCH code with distance $2t+1$, we find the $m$ used to construct the code as the smallest $m$ such that $2^{m}-1-mt>\ell$, and we can see $m$ is in $O(\log\ell)$. The number of additional stabilizer measurements required when using a BCH code as a syndrome measurement code is the difference between the number of encoded and data bits in the BCH code. This is definitionally $R(m,t)\leq mt$, so $R(m,t)\in O(mt)\implies R(m,t)\in O(t\log\ell)$. This $R(m,t)$ is the same for a BCH code and any resulting shortened BCH code, so our methodology requires $O(t\log\ell)$ additional stabilizer measurements.
\end{proof}

The significance of this improvement can be seen by considering the rotated surface code. A $\qcode{d^2,1,d}$ rotated surface code is defined by $\ell=d^2-1$ stabilizers; this means that $\ell\in O(t^2)$ for $t=\floor{\frac{d-1}{2}}$. A standard way to protect against $t$ syndrome errors, proposed by Shor, requires at least $2t$ and up to $t^2$ additional rounds of syndrome extraction \cite{Shor1996}. A minimum-weight perfect-matching (MWPM) space-time decoder is the standard method for fault tolerant syndome extraction for the surface code \cite{DennisJMP2002, FowlerPRL2012}. The method uses $d=2t+1\in O(t)$ rounds of syndrome measurements, yielding $O(\ell t)$ total measurements . In terms of $t$, this means that the most efficient $t$-fault-tolerant syndrome extraction for the surface code requires  $O(t^3)$ additional measurements. Fujiwara's methodology requires $O(t^3\log\ell)$ additional measurements, which is $O(t^3\log t)$ in terms of $t$. In contrast, using our methodology, we require $O(t\log\ell)$ additional measurements; in terms of $t$ this means our method is $O(t\log t)$ for the surface code, significantly better than the alternatives in terms of measurements, but at the cost of a complex syndrome extraction circuit relative to MWPM.
    \subsection{Example: encoding the $\qcode{7,1,3}$ Steane code in a BCH code}
    The Steane code is a CSS code that encodes $X$ and $Z$ errors using the $[7,4,3]$ Hamming code. This Hamming code has parity check matrix:
    $$H_{H}=\mqty[1&0&1&0&1&0&1\\0&1&1&0&0&1&1\\0&0&0&1&1&1&1].$$
    So the  generators of the Steane code's stabilizer group $\mathcal{S}$ and corresponding binary parity-check matrix $H_{S}$ are:
    $$\mathcal{S}=\left<XIXIXIX,IXXIIXX,IIIXXXX,\right.$$$$\left.ZIZIZIZ,IZZIIZZ,IIIZZZZ\right>,$$
    
    \setcounter{MaxMatrixCols}{30}
    $$H_{S}=\mqty[0&H_{H}\\H_{H}&0].$$

    This code gives us a syndrome of length 6. Say we want to protect against $t_{S}=3$ possible syndrome errors. Then we determine our relevant BCH code. The smallest $m$ such that $6<2^{m-1}-3m-1$ is $m=5$, so our BCH code is a $[31,16,7]$ code. We need to shorten it by $a=10$ bits in order for it to encode exactly 6 bits; the resulting code is a $[21,6,7]$ code with generator matrix $G_{B}$:
    
    $$G_{B}=
    \smqty[
1 & 0 & 0 & 0 & 0 & 0 & 0 & 0 & 0 & 1 & 0 & 1 & 0 & 1 & 1 & 0 & 1 & 0 & 0 & 1 & 0 \\
0 & 1 & 0 & 0 & 0 & 0 & 0 & 0 & 0 & 0 & 1 & 0 & 1 & 0 & 1 & 1 & 0 & 1 & 0 & 0 & 1 \\
0 & 0 & 1 & 0 & 0 & 0 & 1 & 1 & 1 & 1 & 0 & 0 & 0 & 0 & 1 & 0 & 0 & 1 & 1 & 0 & 0 \\
0 & 0 & 0 & 1 & 0 & 0 & 0 & 1 & 1 & 1 & 1 & 0 & 0 & 0 & 0 & 1 & 0 & 0 & 1 & 1 & 0 \\
0 & 0 & 0 & 0 & 1 & 0 & 0 & 0 & 1 & 1 & 1 & 1 & 0 & 0 & 0 & 0 & 1 & 0 & 0 & 1 & 1 \\
0 & 0 & 0 & 0 & 0 & 1 & 1 & 1 & 1 & 0 & 1 & 0 & 1 & 1 & 1 & 1 & 1 & 0 & 0 & 0 & 1].$$

    Multiplying $G_{B}^T$ by our matrix $H_{S}$ gives us a $21\times 14$ matrix that defines a set of stabilizers whose measurement returns a length-21 syndrome vector. This can be decoded by appending $a=10$ zeroes to the beginning of the vector and passing it through the decoder for the original $[31,16,7]$ code (for more on the decoding of cyclic codes see \cite{Bose1960a}). The resulting length-16 vector will begin with 10 zeros that can be removed, and we are left with a 6-bit syndrome---which can be used to identify and correct errors on our 7 data qubits in the usual way.

    \subsection{Time overhead}

    A QDS code that uses a BCH code as an SM code will need only $O(mt_{S})$ additional measurements to be performed. Note that $m\in O(\log\ell)$, so the number of additional measurements is $O(t_{S}\log\ell )$. Compared to Fujiwara's construction, which is $O(t_c^3\log\ell)$, we can see that ours gives a significant improvement. Note that the $t_c$ in Fujiwara's construction and the $t_S$ in ours are not necessarily the same. Fujiwara's construction restricts $t_c$ to be at most the $t$ of the quantum code the QDS code is based on. If we desire to protect against a number of syndrome errors fewer than those corrected by the base code, our construction outperforms. If we instead protect against more syndrome errors than $t$, our construction allows for this while Fujiwara's does not. Specifically, if we have a quantum code and want to protect against fewer than $t$ syndrome errors, Fujiwara's construction grows with the number of errors cubed. For the same code and desired correctable syndrome errors, our BCH construction requires a number of additional stabilizer measurements that is linear in $t$. For a comparison of the additional stabilizer measurements needed to obtain the same error-correcting properties between our and Fujiwara's constructions see Fig. \ref{fig:comparison2}. 

    \begin{figure}
        \centering
        \includegraphics[width=1\linewidth]{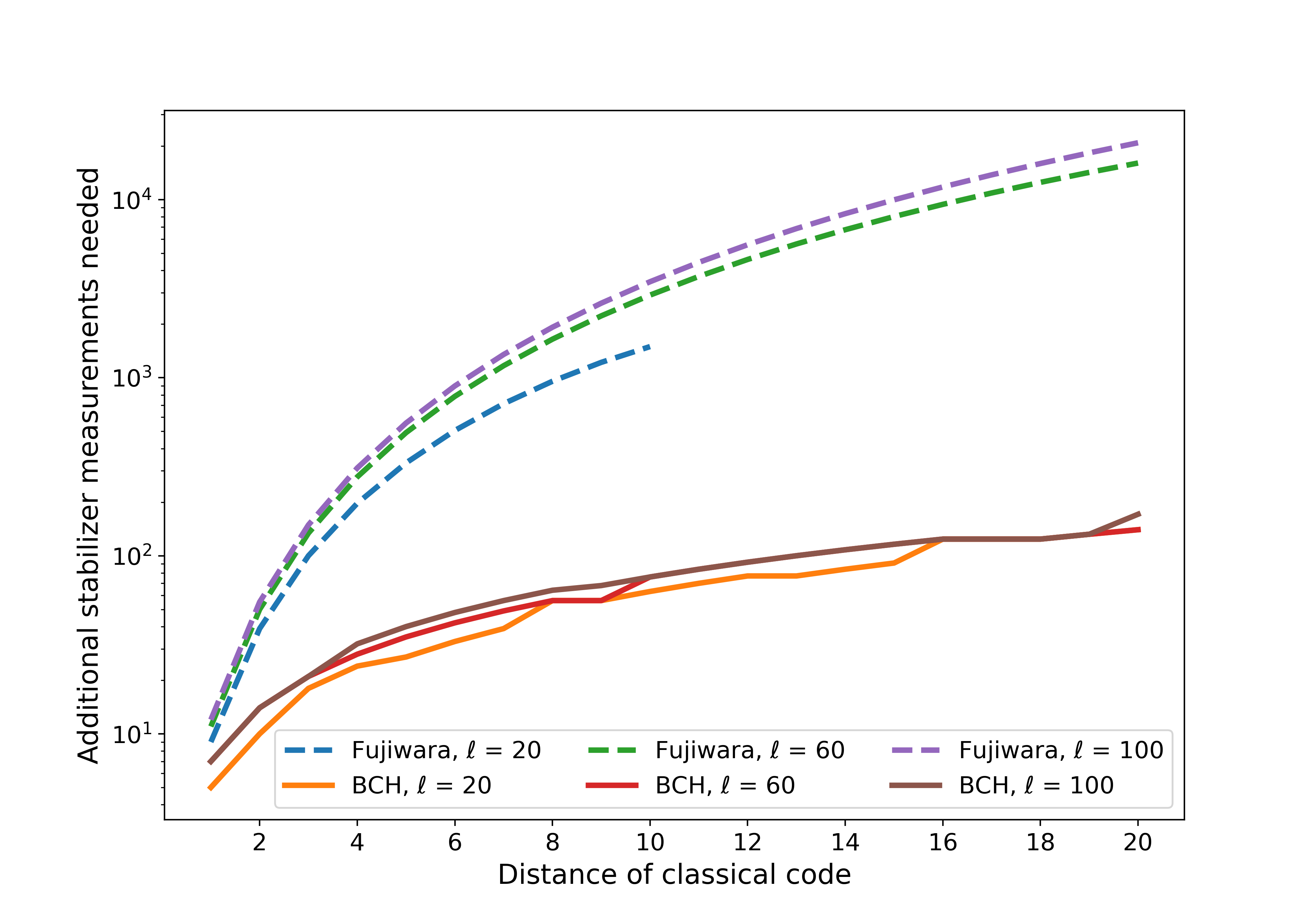}
        \caption{Comparison between the number of additional measurements required by our BCH encoding vs. the additional measurements required by Fujiwara's construction for several values of $\ell=n-k$. Note that the latter's construction is also limited by the number of stabilizer generators in the quantum code.}
        \label{fig:comparison2}
    \end{figure}

    This means that protecting against a specific number of syndrome errors can be done with significantly less time overhead using our encoding procedure than using Fujiwara's construction. Perhaps more meaningfully, it also means that if a circuit has a limited amount of time to perform stabilizer measurements, our construction can correct significantly more errors on the syndrome bits than Fujiwara's construction.

    As an example of this, consider a code with $\ell=10$. If we want to protect against up to 3 syndrome errors, Fujiwara's construction requires up to 76 additional stabilizers. The highest-distance BCH code encoding $\ell=10$ bits such that $R(m,t)\leq 76$ is the $[80,10,23]$ shortened BCH code (shortened from a $[127,57,23]$ code) that only needs 70 additional measurements while protecting against up to 12 syndrome errors. Our methodology allows for significantly more well-protected syndrome measurement in the same amount of time.

    \section{Comparisons}
    To illustrate this, we perform Monte Carlo simulation over many error weights. To obtain accurate results despite the need for high-weight errors with low probability, we determine the probabilities of both logical and decoded syndrome errors for each pair of syndrome error weight $w_s$ and qubit error weight $w_q$, a methodology outlined in \cite{Gutierrez2019}. For a system with probability of syndrome error $p_s$ on a number of syndrome bits $\ell$, and probability of qubit error $p_q$ on number of qubits $q$, we can calculate the probability of a logical error:
    $$p_{err}(p_q,p_s)=\sum_{w_s=1}^{\ell}\sum_{w_q=1}^qA_{w_q,w_s}(p_q,p_s)p_L(w_q,w_s),$$
    where $A_w(p,n)=\binom{n}{w}(p)^{w}(1-p)^{q-n}$ is the probability that $w$ errors will occur on $n$ possible locations when an error occurs with probability $p$, and $A_{w_q,w_s}(p_q,p_s)=A_{w_q}(p_q,q)*A_{w_s}(p_s,\ell)$ is the probablity that exactly $w_q$ qubit errors and $w_s$ syndrome errors will occur. Then $p_L(w_q,w_s)$ is the Monte Carlo-determined probability of an error resulting from the decoding process when $w_q$ qubit errors and $w_s$ syndrome errors occur. Many of the smallest terms of this can be truncated, and the terms where both weights are below the capabilities of the code are guaranteed to be zero, so it is computationally less intense than traditional simulation.

    If we take our encoding and Fujiwara's encoding, in similar numbers of additional stabilizer measurements, our construction performs significantly better under a wide range of possible measurement error probabilities. As an example, consider this encoding of an $\ell=10$ code (Fig. \ref{fig:comparison1}). The probability of a set of bit flips resulting in a logical error is significantly lower with the BCH SM code at all bit-flip error probabilities, with a very pronounced effect at lower probabilities due to the significantly higher distance in the BCH construction.

    \begin{figure}
        \centering
        \includegraphics[width=\linewidth]{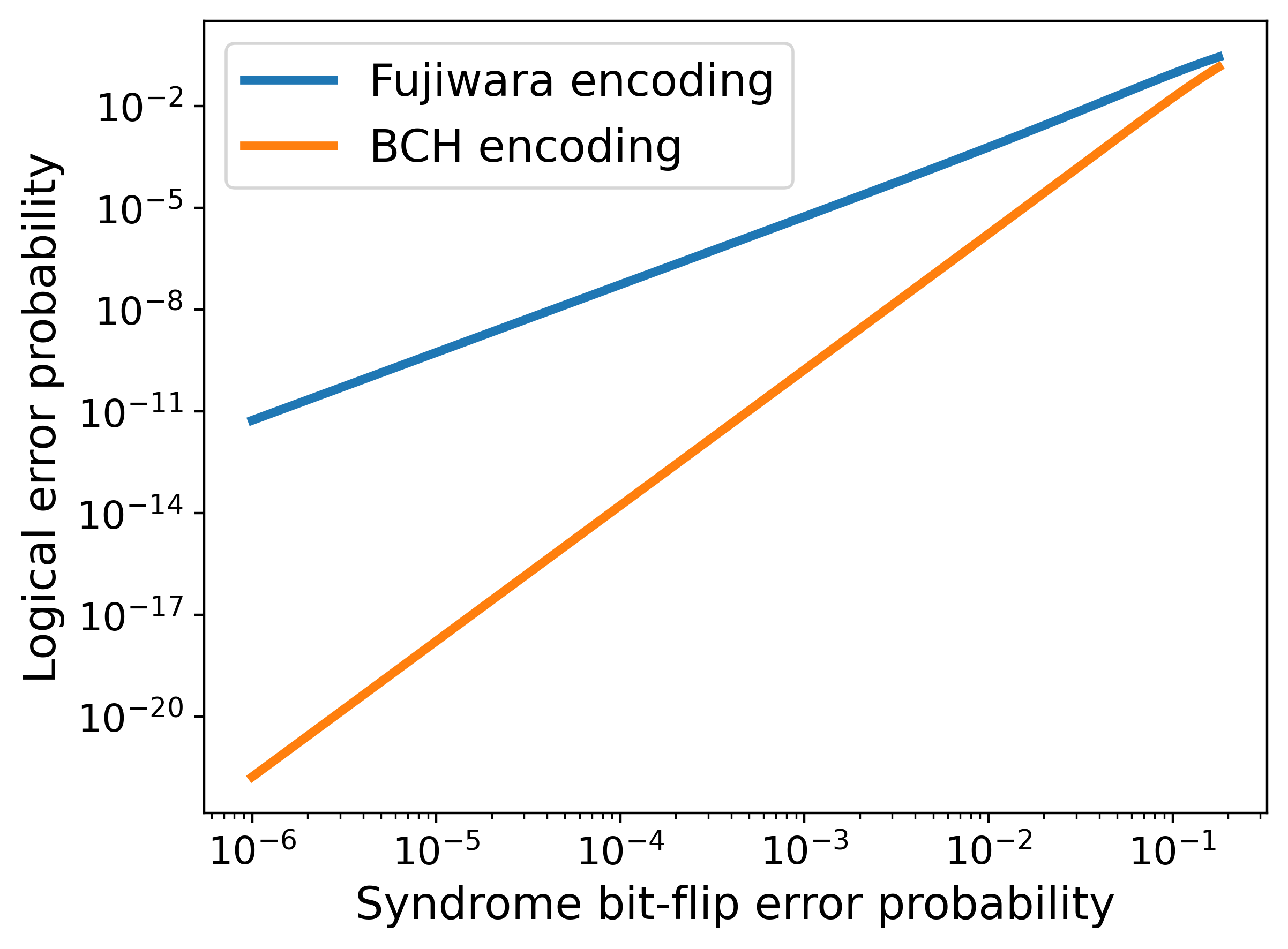}
        \caption{Behavior of BCH and Fujiwara encoding on a code with $\ell=10$, $t_{c} = 3$, $t_{BCH}=12$, and the same number of additional stabilizer measurements. This shows the probability of a logical error arising from incorrect stabilizer decoding. Errors in the results due to sample error from the Monte Carlo simulations are negligible.}
        \label{fig:comparison1}
    \end{figure}

    \begin{figure}
        \centering
        \includegraphics[width=\linewidth]{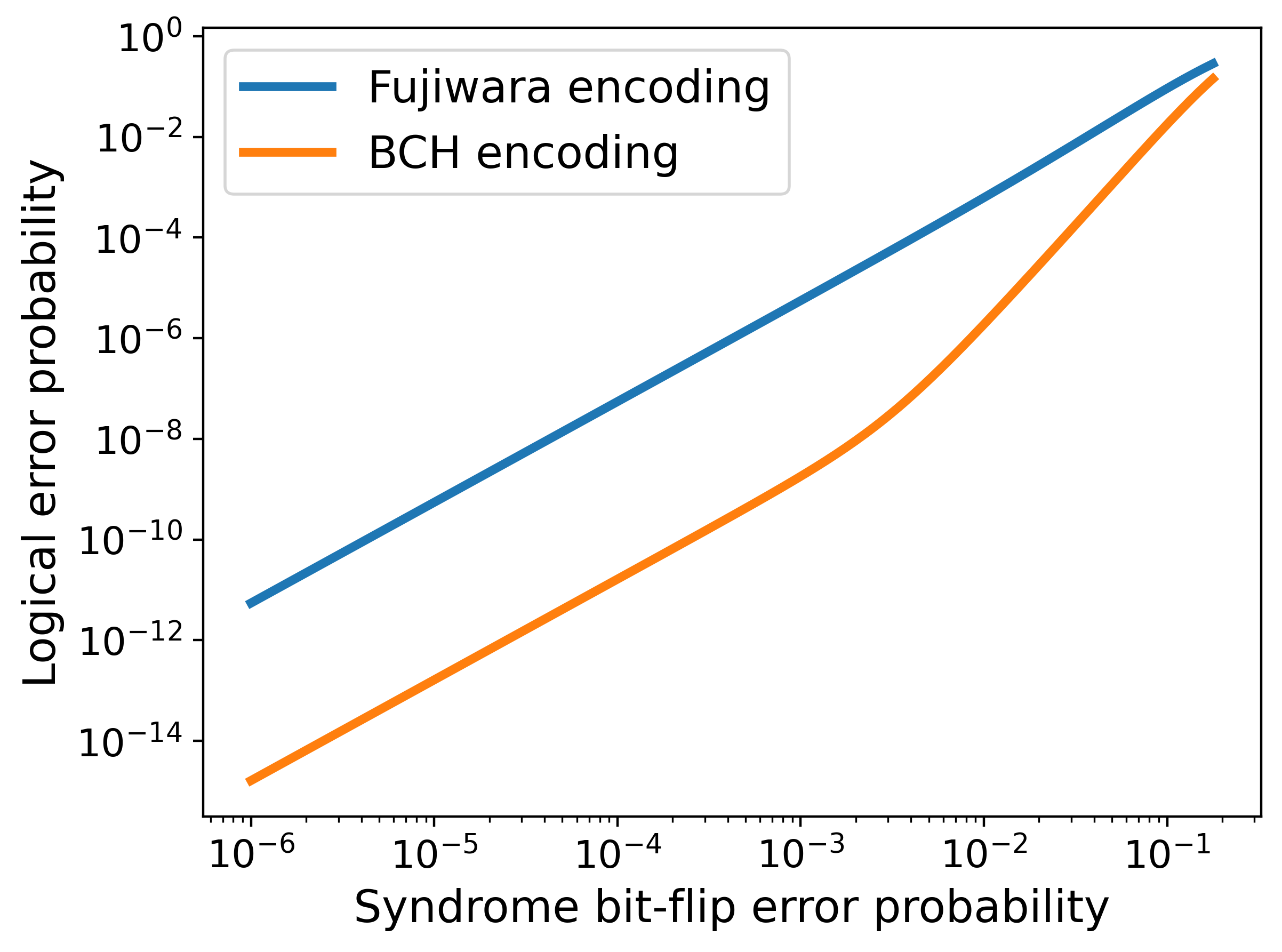}
        \caption{Behavior of BCH and Fujiwara encoding on the Steane code when qubit error probability is 100x less that syndrome bit-flip error probability. This shows that at very high error probabilities, both codes perform relatively similarly; however the BCH code performs significantly better at lower error probabilities. Errors in the results due to sample error from the Monte Carlo simulations are negligible.}
        \label{fig:sameprob}
    \end{figure}

    In a realistic model, it is likely that qubit errors and syndrome bit-flip error probabilities are similar to each other. It can be helpful then to look at the behavior of these codes while varying both errors. In Fig. \ref{fig:sameprob} we look at an encoding of the Steane code in a BCH code and in Fujiwara's construction for syndrome bit-flip probability 100x greater than the Pauli error probability. We can see that the codes both reach the same slope, because the probability of qubit error becomes more significant than the probability of measurement error. However, because the BCH code has a higher distance, it behaves significantly better.
 
	\section{Conclusion}
    We have shown that it is possible to create quantum error-correcting codes that are robust against syndrome measurement errors. We show that using the classical BCH family of cyclic binary codes, we can protect against any number of syndrome errors $t$ we choose, while only requiring a number of additional stabilizer measurements linear in $t$. This is a significant improvement on the previous construction given in \cite{Fujiwara2015}, which requires a number of additional measurements cubic in $t$. Our use of BCH codes confers two advantages. First, for a specific desired number of syndrome errors to correct, our construction requires significantly fewer syndrome measurements and therefore takes significantly less time than that in \cite{Fujiwara2015}. Second, if we have a certain amount of time allocated for syndrome measurement, our construction allows for significantly more syndrome errors to be corrected with the same time overhead as Fujiwara's method.

    One direction of future research is investigating how the use of BCH codes as syndrome measurement codes can be combined with other methods of fault-tolerant techniques such as flag error correction \cite{Chao2020,Anker2022,Tansuwannont2023}, adaptive syndrome measurement \cite{Tansuwannont2023,Pato2023}, and single-shot error correction \cite{Delfosse2022,Bombin2015,Campbell2019,Lin2023}, and whether this improves their error-correcting properties. Another direction for future study is to go beyond the phenomenological error model and investigate the usage of SM codes under full circuit noise. 

    Although BCH codes are generally good codes, this methodology may not the optimal choice for specific quantum codes. For highly structured families of quantum codes, using BCH codes as SM codes may not preserve the structure of the underlying quantum code. For example, a good classical code may not be a good SM code for a QLDPC code if it measures high-weight stabilizer elements. For surface codes, using a BCH code as a SM code likely will not preserve the locality of stabilizer elements. For a CSS code it may be desirable to avoid having to measure the product of both $X$- and $Z$-type stabilizers, and therefore to encode the $X$- and $Z$-type subgroups separately with SM codes. Therefore, there is significant work that is needed for determining the best SM codes for these families. Future work may therefore include encoding subgroups of the stabilizer generators in separate syndrome measurement codes to preserve certain properties of the base quantum code.
\section*{Acknowledgements}
This work was supported by the ARO/LPS QCISS program (W911NF-21-1-0005) and the NSF QLCI for Robust Quantum Simulation (OMA-2120757). Support is also acknowledged from the U.S. Department of Energy, Office of Science, National Quantum Information Science Research Centers, Quantum Systems Accelerator.
    
	\nopagebreak
	
	\thispagestyle{empty}
	
\printbibliography

@article{Fujiwara2014,
Author = {Yuichiro Fujiwara},
Title = {Ability of stabilizer quantum error correction to protect itself from its own imperfection},
Year = {2014},
Eprint = {arXiv:1409.2559},
Howpublished = {Physical Review A, 90 (2014) 062304},
Doi = {10.1103/PhysRevA.90.062304},
}

@article{Anker2022,
author = {Benjamin Anker and Milad Marvian},
title = {{Flag Gadgets based on Classical Codes}},
note = {arXiv:2212.10738 [quant-ph]},
year = {2022}
}

@inproceedings{Ashikhmin2014,
author = {Alexei Ashikhmin and Ching-Yi Lai and Todd A. Brun},
title = {{Robust Quantum Error Syndrome Extraction by Classical Coding}},
booktitle = {Proceedings of the 2014 IEEE International Symposium on Information Theory (ISIT)},
address = {Honolulu, Hawaii, USA},
month = {Jun.},
year = {2014},
pages = {546-550}
}

@inproceedings{Ashikhmin2016,
author = {Alexei Ashikhmin and Ching-Yi Lai and Todd A. Brun},
title = {{Correction of Data and Syndrome Errors by Stabilizer Codes}},
booktitle = {Proceedings of the 2016 IEEE International Symposium on Information Theory (ISIT)},
address = {Barcelona, Spain},
month = Jul,
year = {2016},
pages = {2274-2278}
}

@article{Ashikhmin2020,
author = {Alexei Ashikhmin and Ching-Yi Lai and Todd A. Brun},
title = {{Quantum Data Syndrome Codes}},
journal = {IEEE Journal on Selected Areas in Communications},
volume = {38},
number = {3},
year = {2020},
pages = {449-462}
}

@article{Bombin2015,
author = {H{\'e}ctor Bomb{\'i}n},
title = {{Single-Shot Fault-Tolerant Quantum Error Correction}},
journal = {Phys. Rev. X},
volume = {5},
number = {3},
year = {2015},
pages = {031043}
}

@article{Bose1960a,
author = {R. C. Bose and D. K. Ray-Chaudhuri},
title = {{On A Class of Error Correcting Binary Group Codes}},
journal = {Information and Control},
volume = {3},
number = {1},
year = {1960},
pages = {68-79}
}

@article{Bose1960b,
author = {R. C. Bose and D. K. Ray-Chaudhuri},
title = {{Further Results on Error Correcting Binary Group Codes}},
journal = {Information and Control},
volume = {3},
number = {3},
year = {1960},
pages = {279-290}
}

@article{Campbell2019,
author = {Earl T. Campbell},
title = {{A theory of single-shot error correction for adversarial noise}},
journal = {Quantum Science and Technology},
volume = {4},
number = {2},
year = {2019},
pages = {025006}
}

@article{DiVincenzo1996,
author = {David P. DiVincenzo and Peter W. Shor},
title = {{Fault-Tolerant Error Correction with Efficient Quantum Codes}},
journal = {Physical Review Letters},
volume = {77},
number = {15},
year = {1996},
pages = {3260-3263}
}

@inproceedings{Fujiwara2015,
author = {Yuichiro Fujiwara},
title = {{Global Stabilizer Quantum Error Correction with Combinatorial Arrays}},
booktitle = {Proceedings of the 2015 IEEE International Symposium on Information Theory (ISIT)},
address = {Hong Kong, China},
month = {Jun.},
year = {2015},
pages = {1114-1118}
}

@article{Hocquenghem1959,
author = {Alexis Hocquenghem},
title = {{Codes correcteurs d'erreurs}},
journal = {Chiffres},
volume = {2},
year = {1959},
pages = {147-156}
}

@inproceedings{Kuo2021,
author = {Kao-Yueh Kuo and I-Chun Chern and Ching-Yi Lai},
title = {{Decoding of Quantum Data-Syndrome Codes via Belief Propagation}},
booktitle = {Proceedings of the 2021 IEEE International Symposium on Information Theory (ISIT)},
address = {Melbourne, Australia},
month = {Jul.},
year = {2021},
pages = {1552-1557}
}

@article{Nemec2023,
author = {Andrew Nemec},
title = {{Quantum Data-Syndrome Codes: Subsystem and Impure Code Constructions}},
note = {arXiv:2302.01527 [quant-ph]},
year = {2023}
}

@article{Premakumar2021,
author = {Vickram N. Premakumar and Hele Sha and Daniel Crow and Eric Bach and Robert Joynt},
title = {{2-designs and redundant syndrome extraction for quantum error correction}},
journal = {Quantum Information Processing},
volume = {20},
number = {3},
year = {2021},
pages = {84}
}

@inproceedings{Raveendran2022,
author = {Nithin Raveendran and Narayanan Rengaswamy and Asit Kumar Pradhan and Bane Vasi{\'c}},
title = {{Soft Syndrome Decoding of Quantum LDPC Codes for Joint Correction of Data and Syndrome Errors}},
booktitle = {Proceedings of the 2022 IEEE International Conference on Quantum Computing and Engineering (QCE)},
address = {Broomfield, Colorado, USA},
month = {Sep.},
year = {2022},
pages = {275-281}
}

@article{Shor,
  title = {Scheme for reducing decoherence in quantum computer memory},
  author = {Shor, Peter W.},
  journal = {Phys. Rev. A},
  volume = {52},
  issue = {4},
  pages = {R2493--R2496},
  numpages = {0},
  year = {1995},
  month = {Oct},
  publisher = {American Physical Society},
  doi = {10.1103/PhysRevA.52.R2493},
}

@inproceedings{Shor1996,
author = {P. W. Shor},
title = {{Fault-tolerant quantum computation}},
booktitle = {Proceedings of the 37th IEEE Symposium on Foundations of Computer Science (FOCS)},
address = {Burlington, Vermont, USA},
month = {Oct.},
year = {1996},
pages = {56-67}
}

@article{Calderbank1997,
      title={Quantum Error Correction via Codes over GF(4)}, 
      author={A. R. Calderbank and E. M Rains and P. W. Shor and N. J. A. Sloane},
      year={1997},
      eprint={quant-ph/9608006},
      archivePrefix={arXiv},
      primaryClass={quant-ph}
}

@article{Delfosse2022,
	doi = {10.1109/tit.2021.3120685},
	year = 2022,
	month = {jan},
	publisher = {Institute of Electrical and Electronics Engineers ({IEEE})},
	volume = {68},
	number = {1},
  
	pages = {287--301},
  
	author = {Nicolas Delfosse and Ben W. Reichardt and Krysta M. Svore},
  
	title = {Beyond Single-Shot Fault-Tolerant Quantum Error Correction},
  
	journal = {{IEEE} Transactions on Information Theory}
}

@book{Abramowitz1972, place={Washington, D.C.}, title={Handbook of Mathematical Functions with formulas, graphs, and mathematical tables}, publisher={U.S. Dept. of Commerce, National Bureau of Standards}, author={Abramowitz, Milton and Stegun, Irene A.}, year={1972}}

@article{Tansuwannont2023,
	doi = {10.22331/q-2023-08-08-1075},

  
	year = 2023,
	month = {aug},
  
	publisher = {Verein zur Forderung des Open Access Publizierens in den Quantenwissenschaften},
  
	volume = {7},
  
	pages = {1075},
  
	author = {Theerapat Tansuwannont and Balint Pato and Kenneth R. Brown},
  
	title = {Adaptive syndrome measurements for Shor-style error correction},
  
	journal = {Quantum}
}

@article{Pato2023,
      title={Optimization tools for distance-preserving flag fault-tolerant error correction}, 
      author={Balint Pato and Theerapat Tansuwannont and Shilin Huang and Kenneth R. Brown},
      year={2023},
      eprint={2306.12862},
      archivePrefix={arXiv},
      primaryClass={quant-ph}
}

@article{Gutierrez2019,
	doi = {10.1103/physreva.99.022330},
  
	url = {https://doi.org/10.1103%2Fphysreva.99.022330},
  
	year = 2019,
	month = {feb},
  
	publisher = {American Physical Society ({APS})},
  
	volume = {99},
  
	number = {2},
  
	author = {M. Guti{\'{e}}rrez and M. Müller and A. Berm{\'{u}}dez},
  
	title = {Transversality and lattice surgery: Exploring realistic routes toward coupled logical qubits with trapped-ion quantum processors},
  
	journal = {Physical Review A}
}

@inproceedings{Zeng2019,
author = {Weilei Zeng and Alexei Ashikhmin and Michael Woolls and Leonid P. Pryadko},
title = {{Quantum convolutional data-syndrome codes}},
booktitle = {Proceedings of the IEEE 20th International Workshop on Signal Processing Advances in Wireless Communications (SPAWC)},
address = {Cannes, France},
month = {Jul.},
year = {2019},
pages = {}
}

@article{Chao2020,
  title = {Flag Fault-Tolerant Error Correction for any Stabilizer Code},
  author = {Chao, Rui and Reichardt, Ben W.},
  journal = {PRX Quantum},
  volume = {1},
  issue = {1},
  pages = {010302},
  numpages = {6},
  year = {2020},
  month = {Sep},
  publisher = {American Physical Society},
  doi = {10.1103/PRXQuantum.1.010302},
  url = {https://link.aps.org/doi/10.1103/PRXQuantum.1.010302}
}

@article{Lin2023,
      title={Single-shot error correction on toric codes with high-weight stabilizers}, 
      author={Yingjia Lin and Shilin Huang and Kenneth R. Brown},
      year={2023},
      eprint={2310.16160},
      archivePrefix={arXiv},
      primaryClass={quant-ph}
}

@article{FowlerPRL2012,
  title = {Proof of Finite Surface Code Threshold for Matching},
  author = {Fowler, Austin G.},
  journal = {Phys. Rev. Lett.},
  volume = {109},
  issue = {18},
  pages = {180502},
  numpages = {4},
  year = {2012},
  month = {Nov},
  publisher = {American Physical Society},
  doi = {10.1103/PhysRevLett.109.180502},
  url = {https://link.aps.org/doi/10.1103/PhysRevLett.109.180502}
}

@article{DennisJMP2002,
    author = {Dennis, Eric and Kitaev, Alexei and Landahl, Andrew and Preskill, John},
    title = "{Topological quantum memory}",
    journal = {Journal of Mathematical Physics},
    volume = {43},
    number = {9},
    pages = {4452-4505},
    year = {2002},
    month = {08},
    issn = {0022-2488},
    doi = {10.1063/1.1499754},
    url = {https://doi.org/10.1063/1.1499754},
    eprint = {https://pubs.aip.org/aip/jmp/article-pdf/43/9/4452/8171926/4452\_1\_online.pdf},
}
	
\end{document}